\documentclass[12pt]{iopart}

\usepackage{iopams}
\usepackage[T1]{fontenc}
\usepackage{hyperref}
\usepackage{graphicx}
\usepackage{listings}
\usepackage{enumerate}
\usepackage{cite}
\usepackage{url}
\usepackage{multirow}
\usepackage{adjustbox}
\usepackage{geometry}
\usepackage{framed}
\usepackage{ifthen}

\usepackage{tikz}

\newcommand{\ie}{{\emph{i.e.\/}}}

\newcommand{\C}{\ensuremath{\mathbb{C}}}

\newcommand{\ket}[1]{\ensuremath{|#1\rangle}}
\newcommand{\bra}[1]{\ensuremath{\langle#1|}}
\newcommand{\ketbra}[2]{\ensuremath{\ket{#1} \! \bra{#2}}}
\newcommand{\proj}[1]{\ensuremath{\ketbra{#1}{#1}}}
\newcommand{\braket}[2]{\ensuremath{\langle{#1}|{#2}\rangle}}

\newcommand{\Id}{{\rm 1\hspace{-0.9mm}l}}

\newcommand{\XX}{\mathcal{X}}

\newcommand{\YY}{\mathcal{Y}}

\newcommand{\PP}{\mathcal{P}}

\newcommand{\eqref}[1]{(\ref{#1})}

\newtheorem{proposition}{Proposition}

\newcommand{\text}[1]{\textrm{#1}}

\usepackage{lineno}

\def\>{\rangle}
\def\<{\langle}

\makeatletter
\newcommand{\coloneqq}{\mathrel{\mathop:}=}
\makeatother

\newcommand{\orcid}[1]{\href{https://orcid.org/#1}{\includegraphics[width=12pt]{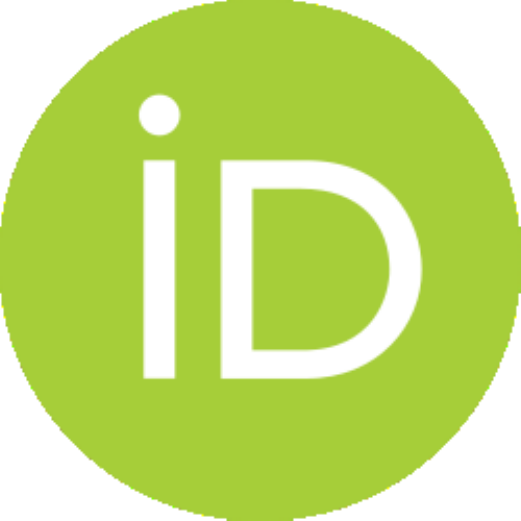}}}

\begin{document}

\title[Benchmarking gate-based quantum devices]{Benchmarking gate-based quantum devices via certification of qubit von Neumann measurements}
\author{Paulina Lewandowska$^{1}$\orcid{0000-0003-1564-7782}, Martin Beseda$^{2}$,*\orcid{0000-0001-5792-2872}}
\address{$^1$ IT4Innovations, VSB~-~Technical University of Ostrava, 17.~listopadu 2172/15, 708 33 Ostrava, Czech Republic}
\address{$^2$ Dipartimento di Ingegneria e Scienze dell'Informazione e Matematica, Universit\`{a} dell'Aquila, via Vetoio, I-67010 Coppito-L'Aquila, Italy}
\address{$^*$ \textit{Corresponding author. E-mail: martin.beseda@univaq.it}}

\begin{abstract}
We present an updated version of PyQBench, an open-source Python library designed for benchmarking gate-based quantum computers, with a focus on certifying qubit von Neumann measurements. This version extends PyQBench's capabilities by incorporating a certification scheme of quantum measurements that evaluates the accuracy on Noisy Intermediate-Scale Quantum devices, alongside its original functionality of von Neumann measurements' discrimination. PyQBench offers a user-friendly command-line interface and Python library integration, allowing users to implement custom measurement schemes and error models for more advanced benchmarking tasks. The new version is specifically designed to support IBM Q devices through integration with the Qiskit library, enabling seamless benchmarking on real quantum hardware. By providing flexible benchmarking tools, PyQBench addresses the critical need for reliable performance metrics in the NISQ era, contributing to the development of error mitigation techniques and the verification of quantum measurement fidelity. The source code is available on GitHub under an open-source license, encouraging community collaboration and further advancements in quantum hardware benchmarking. 
\end{abstract}

\noindent \textbf{Keywords.} 
	Quantum computing,
	Benchmarking,
	Certification, Open-source,
	Python, IBM Q devices, Qiskit

\maketitle
 
\section{Introduction}
Noisy Intermediate-Scale Quantum (NISQ) devices, as outlined by Preskill \cite{preskill2018quantum}, are rapidly developing, built on a variety of architectures and supported by a growing ecosystem of software solutions. Public access to gate-based quantum devices is provided by several hardware vendors, including Rigetti \cite{rigetti}, IBM \cite{ibmq}, IonQ \cite{ionq}, and Xanadu \cite{xanadu}. Each of these vendors typically offers its own software stack and application programming interfaces (APIs) for interacting with their devices. For example, IBM Q devices can be accessed via the Qiskit library \cite{qiskit} or through the IBM Quantum Experience web interface \cite{ibmqplatform}. However, as NISQ devices are still in the early stages of development, they are often characterized by significant noise, errors, and limited qubit counts, which pose challenges for the reliability of quantum computations. Consequently, there is an increasing need for robust benchmarking methodologies to assess their performance.

Benchmarking NISQ devices involves the evaluation of key performance metrics, including quantum gate fidelity, coherence times, error rates, and computational accuracy. These metrics enable researchers and developers to assess how well a quantum device performs under real-world conditions, guiding the optimization of both hardware and algorithms for near-term quantum applications. Given the inherent noise in NISQ devices, benchmarking also focuses on understanding how errors propagate through quantum circuits, the effects of noise on computational outcomes, and the effectiveness of error mitigation techniques \cite{nation2021scalable}.

Several well-established benchmarking techniques are currently employed. One widely used method is randomized benchmarking, which assesses quantum gate error rates by applying random sequences of operations and measuring their success rate \cite{knill2008randomized}. This approach has been implemented in libraries such as Qiskit \cite{qiskit-randomized} and PyQuil \cite{pyquil}. Another significant method is cross-entropy benchmarking, which evaluates the precision of quantum computations by comparing the output of a quantum algorithm with a corresponding classical simulation \cite{cross-entropy-bench}. Additionally, the quantum volume metric has become a key tool for characterizing the overall capacity of a quantum device to solve computational problems. Quantum volume takes into account factors such as qubit count, connectivity, and measurement errors, and can be measured using tools like Qiskit \cite{Cross_2019}.

While these methods provide valuable insights into the performance of quantum devices, they focus primarily on gate-level errors or device-specific characteristics. Recently, an alternative benchmarking approach was introduced in \cite{jalowiecki2023pyqbench}, focusing on the certification of qubit von Neumann measurements. This approach offers a more direct evaluation of a device's ability to accurately discriminate between quantum measurements. The authors developed PyQBench, an open-source Python library for benchmarking gate-based quantum computers, which assesses the ability of NISQ devices to perform qubit von Neumann measurement discrimination \cite{puchala2018strategies}. PyQBench supports any device available through the Qiskit library, allowing users to benchmark hardware from providers such as IBM Q \cite{ibmq} and Amazon Braket \cite{amazon}. The library offers a command-line interface (CLI) for running predefined benchmarking scenarios based on the parametrized Fourier family of qubit von Neumann measurements.

In this paper, we build upon the PyQBench framework by incorporating a certication scheme for qubit von Neumann measurements, as described in \cite{lewandowska2021certification}. We provide a simplified CLI for executing benchmarks  certifying the predefined parametrized Fourier family of measurements. Furthermore, for more advanced use cases, PyQBench can be employed as a Python library, allowing users to define custom measurement schemes or integrate error models for a more comprehensive evaluation. The new certification scheme is specifically designed for IBM Q devices and is integrated with the Qiskit library. Additionally, all the functionalities of PyQBench for the discrimination scheme of von Neumann measurements have also been updated to the new version of Qiskit. 

The source code for PyQBench is publicly available on GitHub\footnote{\url{https://github.com/iitis/PyQBench}} under an open-source license, enabling users to adapt and extend the package for their specific research needs. 
By providing this tool, we aim to facilitate the benchmarking of NISQ devices across a range of quantum computing platforms and contribute to the development of reliable quality metrics for near-term quantum technologies.

This paper is organized as follows. Section \ref{sec:preli} presents the necessary mathematical preliminaries. Section \ref{sec:cert} introduces the binary certification scheme for qubit von Neumann measurements, and Section \ref{sec:reali} discusses its realization on NISQ devices. In Section \ref{sec:cert-fourier}, we present the certification scheme for the parametrized Fourier family of qubit measurements. Details regarding the software functionalities are provided in Section \ref{sec:soft}, where Section \ref{sec:lib} focuses on PyQBench's Python library interface, and Section \ref{sec:cli} details the CLI functionalities. The certification results obtained on IBM Q platform are evaluated and discussed in Section \ref{sec:results} both with and without noise mitigation applied. Finally, Section \ref{sec:conclusion} summarizes the main contributions and results of this work. Technical details and examples of YAML configuration files used in our experiments are included in the appendices.

\section{Mathematical preliminaries}\label{sec:preli}
 Consider two complex Euclidean spaces and denote them by $\XX, \YY$. Let $\mathrm{L}(\XX, \YY)$  be the set of all linear operators of the form $M: \XX \rightarrow \YY$. 
 From now on, let $L(X) \coloneqq L(X , X )$, as a shorthand.
 The set of quantum states, that is, positive semidefinite operators that trace equal to one, will be denoted by $\Omega(\XX)$.
 An operator $U \in \mathrm{L}\left(\XX\right) $ is unitary if it satisfies the equation $U U^\dagger = U^\dagger U = \Id$. 
 The notation $\mathrm{U}\left(\XX\right)$ will be used to denote the set of all unitary operators.
 We will distinguish the special class of unitary matrices that is diagonal. This subset will be denoted by $\mathrm{DU}(\XX)$. 
 
 A general quantum measurement, that is, a positive operator-valued measure (POVM) $\mathcal{P}$ is a finite collection of positive semidefinite operators $\{ E_0, \ldots, E_m \}$ called effects, which sum up to identity, \ie $\sum_{i=0}^m E_i = \Id$. If all the effects are rank-one projection operators, then such a measurement is called a von Neumann measurement. Every von Neumann measurement can be parametrized by a unitary matrix.  If $U$  is a unitary matrix, one can construct a von Neumann measurement $\PP_{U}$ by taking projectors onto its columns. In this case, we say that $\PP_{U}$ is described by the matrix $U$. 
 
 Typically, NISQ devices offer measurements on the computational basis, that is, $U = \Id$. To implement an arbitrary von Neumann measurement $\PP_{U}$, one has to first apply $U^\dagger$ to the system and then measure on the computational basis. It can be simply illustrated in Fig.\ref{vonneumann}.

\begin{figure}[hpt!]
	\centering
	\includegraphics[scale=1.7]{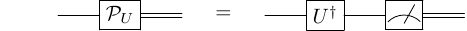}
	\caption{Schematic representation of implementation of a von Neumann measurement using measurement in computational basis. The left circuit shows a symbolic representation of a von Neumann measurement $\PP_{U}$. The right circuit represents its decomposition by changing the basis $U^\dagger$ followed by measurement in the computational basis. } \label{vonneumann}
\end{figure} 
 
Our results often use the terms numerical range and q-numerical range \cite{nr}. The numerical range is a subset of the complex plane defined for a
matrix $X \in \mathrm{L}(\XX)$ as \begin{equation}
W(X) \coloneqq \{ \bra{\phi} X \ket{\phi}: \braket{\phi}{\phi} =1 \}, 
\end{equation}
 while the q-numerical range \cite{tsing1984constrained, li1998q, li1998some}  is defined as 
 \begin{equation}
 W_q(X) \coloneqq \{ \bra{\phi} X \ket{\psi}: \braket{\phi}{\phi} = \braket{\psi}{\psi} = 1, \braket{\phi}{\psi} = q, q\in \C \}. 
 \end{equation}
 We will also use the notation 
 \begin{equation}
 \nu_q(X) \coloneqq \min \{ |x|: x \in W_q(X) \},
 \end{equation}
 as the distance between $q $-numerical range to zero.

 \section{Certification scheme}\label{sec:cert}
Let us consider the following scenario. Imagine we have a black box that contains one of two von Neumann measurements, either $\PP_{\Id}$, a single-qubit measurement in the computational Z-basis, or an alternative measurement $\PP_{U}$, which is performed in the basis $U$. The owner of the box informs us that $\PP_{\Id}$ is the measurement inside, but we do not fully trust this promise. To address this uncertainty, we decided to perform a hypothesis testing scheme. We therefore treat $\PP_\Id$ as the null hypothesis $H_0$ and take $\PP_{U}$ as the alternative hypothesis $H_1$.

Typically, a certification scheme requires an auxiliary qubit \cite{lewandowska2021certification}. Thus, we begin by preparing an input state $\ket{\psi_0}$, which could be entangled. Next, one of the two measurements, $\PP_{\Id}$ or $\PP_{U}$, is performed on the first qubit. After measurement is performed, the null hypothesis $H_0$ corresponds to the state $(\PP_{\Id} \otimes \Id)(\proj{\psi_0})$, while the alternative hypothesis $H_1$ corresponds to the state $(\PP_{U} \otimes \Id)(\proj{\psi_0})$. Based on the outcome $i$ obtained from the first qubit, we perform a final von Neumann measurement $\PP_{V_i}$ on the second one. The outcome $j$ of the measurement $\PP_{V_i}$ determines whether we accept or reject the null hypothesis $H_0$. We assume that if $j=0$, we accept $H_0$; otherwise, if $j=1$, we reject it.

Similarly to classical hypothesis testing, the certification scheme involves two types of errors.  The type I error occurs if we reject the null hypothesis when it was in reality true. The type II error occurs if we accept the null hypothesis, when we should have rejected it. The objective of the certification task is to find an optimal input state and final measurement that minimize one type of error while keeping the other fixed. In this work, our aim is to minimize the type II error, given a fixed type I error. We will assume a statistical significance $\delta \in [0,1]$, that is, the probability of type I error will be upper-bounded by $p_{\text{II}} \le \delta$. A schematic of the certification setup is illustrated in Fig.\ref{fig-cert-real}.

   \begin{figure}[hpt!]
   	\centering
   	\includegraphics[scale=1.7]{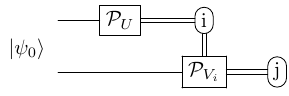}
	   	\caption{Schematic representation of the setup for von Neumann measurement certification scheme with null hypothesis $H_0: \PP_{\Id} \otimes \Id$  and alternative hypothesis $H_1: \PP_{U} \otimes \Id$ with statistical significance $\delta$.  First, the joint system is prepared in some state $\ket{\psi_0}$. Then,
   		the measurement $\PP_{U}$ is performed on the first qubit. Based on the outcome $i$, we perform the final measurement $\PP_{V_i}$ on the second qubit, obtaining the output in $j$. Finally, based on the label $j$ we make a decision. If $j=0$, then we accept the hypothesis $H_0$, otherwise we reject the hypothesis $H_0$. } \label{fig-cert-real}
   \end{figure} 
    
Summarizing, the problem of certification of von Neumann measurements has the following hypotheses
\begin{eqnarray}
H_0&: \PP_\Id \otimes \Id \\
H_1&:  \PP_U \otimes \Id.
\end{eqnarray}
Let us assume the statistical significance $\delta \in [0,1]$.
Using the celebrated result in 
\cite{lewandowska2021certification}, 
one finds that the minimized probability of type II error yields
	\begin{equation}\label{cert-max}
p_{\text{II}} = \max_{E \in \text{DU}(\XX)} \nu^2_{\sqrt{1-\delta}} \left(UE\right).
\end{equation}
But how do we construct the appropriated initial state and final measurement, which minimize the probability of type II error? 
To construct the optimal initial state for certification task between $\PP_{U}$ and $\PP_{\Id}$, one starts by calculating their distance in the notion of the diamond norm \cite{watrous2018theory} given by
\begin{equation}
\| \PP_{U} - \PP_{\Id} \|_\diamond = \max_{\| \ket{\psi}\|_1 = 1 } \| \left((\PP_{U} - \PP_{\Id})  \otimes \Id \right) (\proj{\psi})\|_1.  
\end{equation}
A quantum state $\ket{\psi_0}$, which maximizes the diamond norm, is called a discriminator. This state is precisely an optimal initial state in the certification scheme \cite{lewandowska2021certification}. It is worth noting that this state is also optimal for the problem of discrimination between quantum von Neumann measurements $\PP_{U}$ and $\PP_{\Id}$ \cite{watrous2018theory}. However, determining the discriminator is a problematic task that is only sometimes possible to solve. Moreover, just for some relatively small number of qubits, the diamond norm and the form of discriminator can be computed, e.g., using semidefinite programming (SDP)\cite{watrous2009semidefinite}. Having the form of the discriminator, determining the final measurements $\PP_{V_0}$ and $\PP_{V_1}$ is straightforward using the results from \cite{lewandowska2021certification}.

\section{Realization of certification scheme on actual NISQ devices }\label{sec:reali}

In PyQBench, benchmarks involve experimentally determining the minimized probability of type II error with a fixed statistical significance $\delta$ between two von Neumann measurements conducted by the device under test and then comparing these results to theoretical predictions. Naturally, to achieve a reliable estimation of the underlying probability distribution, we need to repeat the procedure outlined in Fig.~\ref{fig-cert-real} multiple times. 

However,  we are not able to implement the certification scheme one-to-one. Current NISQ devices cannot perform conditional measurements, presenting a significant obstacle to implementing our scheme on real hardware. We overcome this challenge by modifying our scheme to utilize only the components available on current devices. We achieve this through two equivalent approaches: postselection or direct sum.

\subsection{Postselection}

The first idea uses a postselection scheme in Fig.~\ref{fig:postsellection-cert}. Here, instead of preparing the measurement $\PP_{V_i}$ conditioned on the outcome $i$, we run circuits with both $\PP_{V_0}$ and one with $\PP_{V_1}$ and measure both qubits. We then discard the results of the circuit for which the label measured on the first qubit, $i$, does not match the unitary matrix label $k$ as schematically shown in Fig.~\ref{fig:postsellection-cert}. 	
\begin{figure}[htp!]
	\centering 
	\includegraphics[scale=1.7]{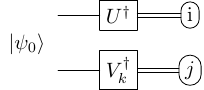} 
	\caption{ A schematic representation of the setup for certificating von Neumann measurements using postselection.
	}\label{fig:postsellection-cert}
\end{figure}  

Hence, our experiments can be grouped into classes identified by tuples of the form $(k, i, j)$, where $k,i,j \in \{0,1\}$, for which $k \ne i$. Hence, the total number of valid experiments is
\begin{equation}
N_\text{total} = \#\{(k, i, j): k = i \}. 
\end{equation}
Finally, we count the valid experiments resulting in successful certification. If we define
\begin{eqnarray}
N_{0} = \#\{(k, i, j): k = i, j = 0\},
\end{eqnarray}
then the empirical probability of type II error yields
\begin{equation}
p_{\text{II}} = \frac{N_{0}}{N_{\text{total}}}. 
\end{equation}

\subsection{Direct sum}
The second idea uses the direct sum $V_0^\dagger \oplus V_1^\dagger  
$ operator. This approach can be described as in Fig.~\ref{fig:controlled-cert}.      Here,
instead of performing a conditional measurement $\PP_{V_i}$, we implement the operator 
\begin{equation}\label{direct}
V_0^\dagger \oplus V_1^\dagger = |0\rangle \langle 0| \otimes V_0^\dagger + |1\rangle \langle 1| \otimes V_1^\dagger.
\end{equation} 
Now, depending on the outcome $i$, one of the summands in Eq.\eqref{direct} vanishes, and we end up performing exactly the same operations as in the original scheme.
\begin{figure}[htp!]
	\centering 
	\includegraphics[scale=1.7]{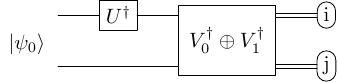} 
	\caption{ A schematic representation of the setup for certificating
		measurements using direct sum   $V_0^\dagger \oplus V_1^\dagger  
		$. }\label{fig:controlled-cert}
\end{figure} 
In this approach, the experiment can be characterized by the tuple $(i,j)$. If we define
\begin{eqnarray}
N_{0} = \#\{(i, j): j = 0\}, 
\end{eqnarray}
then  	the probability of type II error yields 
\begin{equation}
p_{\text{II}} = \frac{N_{0} }{N_{\text{total}}},
\end{equation}
where $N_{\text{total}}$ is the number of trials.

\section{Certification scheme for parametrized Fourier family of measurements} \label{sec:cert-fourier}

In this section, we provide a certification scheme by assuming
that all needed components are known. To do this,  we introduce a parametrized Fourier family of measurements defined as a set $\{ \PP_{U_\phi}: \phi \in [0,2\pi] \}$, where 
\begin{equation} \label{fourier-cert-bench}
U_\phi = H 
\left(\begin{array}{cc}1&0\\0&e^{i \phi}\end{array}\right)  H^\dagger, 
\end{equation} 
and  $H$ is the Hadamard matrix of dimension two. In addition, assume the statistical significance $\delta \in [0,1]$. 
 For each element of this
set, the optimal input state $\ket{\psi_0}$ is the Bell state\begin{equation}
\ket{\psi_0} = \frac{1}{\sqrt{2}} (\ket{00} + \ket{11}),
\end{equation}
whereas the unitaries $V_0$ and $V_1$ have the following form 
		\begin{enumerate}
		\item if $\sqrt{1+\cos\phi} \ge \sqrt{2\delta} $ and $\phi \in [0, \pi)$, then  \begin{equation}
		V_0 = \left(\begin{array}{cc}\sqrt{1-\delta}&\sqrt{\delta}\\-\sqrt{\delta}&\sqrt{1-\delta}\end{array}\right),
		\end{equation} 
		and 
		\begin{equation}
		V_1 = \left(\begin{array}{cc}\sqrt{\delta}&\sqrt{1-\delta}\\\sqrt{1-\delta}&-\sqrt{\delta}\end{array}\right);
		\end{equation} 
		
		\item if $ \sqrt{1+\cos\phi} < \sqrt{2\delta} $ and $\phi \in [0, 2\pi)$, then  \begin{equation}
		V_0 = \left(\begin{array}{cc}\sin \frac{\phi}{2}& | \cos \frac{\phi}{2}|\\-\cos \frac{\phi}{2}&\frac{\sin \phi}{2 |\cos \frac{\phi}{2}| }\end{array}\right),
		\end{equation}
		and 
		\begin{equation}
		V_1 = \left(\begin{array}{cc} | \cos \frac{\phi}{2} | &  \sin \frac{\phi}{2} \\ \frac{\sin\phi}{2|\cos \frac{\phi}{2}|} & - \cos \frac{\phi}{2} \end{array}\right); 
		\end{equation}
		
		\item $\sqrt{1+\cos\phi} \ge \sqrt{2\delta}$ and $\phi \in [\pi, 2\pi)$, then  \begin{equation}
		V_0 = \left(\begin{array}{cc}\sqrt{1-\delta}&-\sqrt{\delta}\\\sqrt{\delta}&\sqrt{1-\delta}\end{array}\right),
		\end{equation}
		and 
		\begin{equation}
		V_1= \left(\begin{array}{cc} - \sqrt{\delta}&\sqrt{1-\delta}\\\sqrt{1-\delta}&\sqrt{\delta}\end{array}\right). 
		\end{equation}
	\end{enumerate} 
Finally,  the minimized  probability of type II error is given by 
\begin{equation}
p_{\text{II}}=\left\{
\begin{array}{ccc}
\left( \frac{ |1+e^{i  \phi} | }{2} \cdot \sqrt{1-\delta} - \sqrt{1-\frac{|1+e^{i  \phi}|^2}{4}}  \cdot \sqrt{\delta } \right)^2 &\mbox{for}&\frac{ |1+e^{i  \phi} | }{2} > \sqrt{\delta},\\
0&\mbox{for}&\frac{ |1+e^{i  \phi} | }{2} \le \sqrt{\delta}.
\end{array}
\right.
\end{equation}
 We explore the construction of the certification strategy for the parametrized Fourier family of measurements in \ref{app:fourier}.

 \section{Software functionalities} \label{sec:soft}
 The new version of PyQBench is dedicated to IBM Q devices available through the Qiskit library. As already described, PyQBench can be used as a Python library or a command line interface (CLI).
Both functionalities are implemented as part of \texttt{qbench} Python package. The exposed CLI tool is also named \texttt{qbench}.  
PyQBench can be installed from the official Python Package Index (PyPI) by running \texttt{pip install pyqbench}. To properly configure the Python environment, simply use YML file (environment.yml). The installation process should also provide the user with the \texttt{qbench} command without the need for further configuration.  For more information on software functions, we can read \cite{jalowiecki2023pyqbench}. 

This section is divided into two parts. In Section~\ref{sec:lib}, we describe PyQBench as Python library. Next, in Section~\ref{sec:cli}, we present general functionality of PyQBench's CLI.

\subsection{PyQBench as Python library} \label{sec:lib}
  In this subsection, we will demonstrate how \texttt{qbench} package can be used with user-defined measurement. When used as a library, PyQBench allows the manipulation of certification scheme by, e.g. adding a noise model. The user then defines a unitary matrix $U$ that describes the measurement to be certified, the optimal initial state $\ket{\psi_0} $ and unitaries $V_0$ and $V_1 $ to create the final measurement. The PyQBench library provides then the following functionalities: assembling circuits for certification scheme, running the whole circuits defining benchmark on specified IBM Q backend and, finally, interpreting the obtained outputs by computing the minimized probability of type II error.

 Inspired by Example 1 in \cite{lewandowska2021certification},   we will present the certification of the measurement performed in the Hadamard basis to show the usage of PyQBench as a Python library. 
The explicit formula for initial state in this case reads
\begin{equation}
\ket{\psi_0} = \frac{1}{\sqrt{2}} (\ket{00} + \ket{11}), 
\end{equation}
  with the final measurements of the forms
\begin{equation}
 		V_0 = \left(\begin{array}{cc}\sqrt{1-\delta}&\sqrt{\delta}\\-\sqrt{\delta}&\sqrt{1-\delta}\end{array}\right),
 		\end{equation} 
and 
\begin{equation}
 		V_1 = \left(\begin{array}{cc}-\sqrt{\delta}&\sqrt{1-\delta}\\\sqrt{1-\delta}&\sqrt{\delta}\end{array}\right),
 		\end{equation} 
      for some choice of the statistical significance $\delta \in [0,1]$. Thus, the minimized probability of type II error is equal to
\begin{equation}
    p_{\text{II}} = \frac{1}{2} (\sqrt{1-\delta } - \sqrt{\delta})^2.
\end{equation}
In our example, let assume $\delta = 0.05$. 
To use the above benchmarking scheme in PyQBench, we first need to construct circuits that can be executed by actual hardware.  The
circuit taking $\ket{00}$ to the Bell state $\ket{\psi_0}$ comprises the Hadamard gate followed by CX gate on both qubits. For $V_0^\dagger$ and $V_1^\dagger$ observe that
$V_0^\dagger = \text{RY}\left( 2 \arcsin\left(\sqrt{0.05}\right) \right)$ and $V_1^\dagger = \text{X} \cdot \text{RY}\left( 2\arcsin\left(\sqrt{0.05}\right) \right)$ where $\text{RY}(\theta)$ is rotation gate around the Y axis defined by \begin{equation}
    \text{RY}\left(\theta \right) = \left(\begin{array}{cc} \cos\left(\frac{\theta}{2} \right)& -\sin\left(\frac{\theta}{2} \right) \\ \sin\left(\frac{\theta}{2} \right)& \cos\left(\frac{\theta}{2} \right) \end{array}\right). 
\end{equation}

We will now demonstrate how to implement this theoretical scheme in PyQBench with some modifications.
For this example we will use the Qiskit Aer simulator \cite{aerstate}. First, we import
the necessary functions and classes from PyQBench and Qiskit.

\begin{lstlisting}[language=Python, caption=Importing the necessary functions and classes]
from qiskit_aer import StatevectorSimulator
from qiskit import QuantumCircuit
import numpy as np
from qbench.schemes.postselection import 
    benchmark_certification_using_postselection
from qbench.schemes.direct_sum import 
    benchmark_certification_using_direct_sum
\end{lstlisting}
To implement the certification scheme in PyQBench, we need to define all
the necessary components as Qiskit instructions. To do so, we first define \texttt{QuantumCircuit(2)} acting on qubit 0 and 1 and then use the \texttt{instruction()} method.

\begin{lstlisting}[language=Python, caption=Defining components for Hadamard experiment]
def state_prep():
    circuit = QuantumCircuit(2)
    circuit.h(0)
    circuit.cx(0,1)
    return circuit.to_instruction()
    
def u_dag():
    circuit = QuantumCircuit(1)
    circuit.h(0)
    return circuit.to_instruction()

def v0_dag():
    circuit = QuantumCircuit(1)
    circuit.ry(2 * np.arcsin(np.sqrt(0.05)), 0)
    return circuit.to_instruction()

def v1_dag():
    circuit = QuantumCircuit(1)
    circuit.ry(2 * np.arcsin(np.sqrt(0.05)), 0)
    circuit.x(0)
    return circuit.to_instruction()

def v0_v1_direct_sum_dag():
    circuit = QuantumCircuit(2)
    circuit.p(-np.pi, 0)
    circuit.ry(-2 * np.arcsin(np.sqrt(0.05)), 0)
    circuit.cx(0, 1)
    return circuit.to_instruction()
\end{lstlisting}
We also need to construct a backend object, which is an instance Aer simulator.
\begin{lstlisting}[language=Python, caption=Defining a backend]
simulator = StatevectorSimulator()
\end{lstlisting}
 \parbox{\textwidth}{Now, when one wishes to run the experiment without any modifications on a given backend, it is enough to run \texttt{benchmark\_certification\_using\_postselection} or \texttt{benchmark\_certification\_using\_direct\_sum} function, depending on the user preference.}

\begin{lstlisting}[language=Python, caption=Simulation benchmark by using postselection]
postselection_result = 
    benchmark_certification_using_postselection(
    backend=simulator,
    target=0,
    ancilla=1,
    state_preparation=state_prep(),
    u_dag=u_dag(),
    v0_dag=v0_dag(),
    v1_dag=v1_dag(),
    num_shots_per_measurement=10000)
\end{lstlisting}

\begin{lstlisting}[language=Python, caption=Simulation benchmark by using direct sum]
direct_sum_result = benchmark_certification_using_direct_sum(
    backend=simulator,
    target=0,
    ancilla=1,
    state_preparation=state_prep(),
    u_dag=u_dag(),
    v0_v1_direct_sum_dag=v0_v1_direct_sum_dag(),
    num_shots_per_measurement=10000)
\end{lstlisting}
The \texttt{postselection\_result} and \texttt{direct\_sum\_result} variables contain now
the empirical probabilities of type II error. We can compare them
with the theoretical value and compute the absolute error.

\begin{lstlisting}[language=Python, caption=Examining the benchmark results]
p_succ = (1/np.sqrt(2) * np.sqrt(0.95) +
        - 1/np.sqrt(2) * np.sqrt(0.05))**2
print(f"Analytical p_succ = {p_succ}")
print(f"Postselection: p_succ = {postselection_result}, 
abs_error ={np.abs(p_succ - postselection_result)}")
print(f"Direct_sum: p_succ = {direct_sum_result}, 
abs_error ={np.abs(p_succ - direct_sum_result)}")
\end{lstlisting}

\begin{lstlisting}[language=Python]
Analytical p_succ = 0.2820550528229661
Postselection: p_succ = 0.28322830780328, 
    abs_error = 0.001173254980313898
Direct_sum: p_succ = 0.28333, abs_error = 0.0012749471770339138
\end{lstlisting}

 But what if we want to modify this process?  PyQBench provides functions performing: assembly of circuits needed for experiment under providing the components and interpretation of the obtained measurements.

For the rest of this example, we focus only on the postselection case, as the
direct sum case is analogous. We continue by importing two more functions
from PyQBench. 

    \begin{lstlisting}[language=Python, caption=Assembling circuits]
from qbench.schemes.postselection import (
    assemble_circuits_certification_postselection,
    compute_probabilities_certification_postselection)

circuits = assemble_circuits_certification_postselection(
    target=0,
    ancilla=1,
    state_preparation=state_prep(),
    u_dag=u_dag(),
    v0_dag=v0_dag(),
    v1_dag=v1_dag())
\end{lstlisting}
 Recall that for the postselection scheme we have two possible choices of a final measurement. The function \texttt{assemble\_circuits\_certification\_postselection} creates two circuits and places them in a dictionary with keys "\texttt{u\_v0}", "\texttt{u\_v1}". Now we will run our circuits using noisy and noiseless simulation. We start by creating a noise model using Qiskit. 

    \begin{lstlisting}[language=Python, caption=Adding noise model]
from qiskit_aer.noise import NoiseModel, ReadoutError

error = ReadoutError([[0.75, 0.25], [0.8, 0.2]])

noise_model = NoiseModel()

noise_model.add_readout_error(error, [0])
noise_model.add_readout_error(error, [1])
\end{lstlisting}
 Now we can execute the circuits with and
without noise. To do this, we will use Qiskit's \texttt{run} function. It should be mentioned that we have to keep track of which measurements correspond to which circuit. We do so by fixing an ordering on the keys in the circuits
dictionary.

    \begin{lstlisting}[language=Python, caption=Running circuits]
keys_ordering = ["u_v0", "u_v1"]

all_circuits = [circuits[key] for key in keys_ordering]

counts_noisy = simulator.run(
    all_circuits,
    backend=simulator,
    noise_model=noise_model,
    shots=100000).result().get_counts()

counts_noiseless = simulator.run(
    all_circuits,
    backend=simulator,
    shots=100000).result().get_counts()
\end{lstlisting}
  Finally, we use the measurement counts to compute discrimination probabilities using \texttt{compute\_probabilities\_certification\_postselection} function.

    \begin{lstlisting}[language=Python, caption=Computing probabilities]
prob_succ_noiseless = 
compute_probabilities_certification_postselection(
    u_v0_counts=counts_noiseless[0],
    u_v1_counts=counts_noiseless[1],)

prob_succ_noisy = 
compute_probabilities_certification_postselection(
    u_v0_counts=counts_noisy[0],
    u_v1_counts=counts_noisy[1],)
\end{lstlisting}
We can now examine the results. From the experiment, we obtain \texttt{prob\_succ\_noiseless} = 0.28072503092454 and \texttt{prob\_succ\_noisy}
= 0.77564942534484. As expected, for noisy simulations, the result lies
further away from the target value of 0.2820550528229661.
This concludes our example. 
  \subsection{PyQBench as \texttt{qbench} CLI} \label{sec:cli}
   PyQBench contains a simplified CLI  for running certification experiments for the parametrized Fourier family of measurements defined previously in Section \ref{sec:cert-fourier}. The CLI configuration is performed by the YAML \cite{yaml} files. One file describes the experiment to be performed,  whereas the second one contains the description of the backend on which the benchmark should be run. 
   In addition, a particular experiment can run in a synchronous or asynchronous mode. The mode of execution is defined in the YAML file describing the backend (see the example below). What does it mean? When we choose to run the experiment in asynchronous mode,   PyQBench submits all tasks performing a benchmark and then writes an intermediate YAML file containing metadata of submitted experiments. The intermediate file can be used to query the status of the submitted jobs and, finally, to resolve them to get the measurement outcomes. In synchronous mode, PyQBench submits all jobs required to run the benchmark and then waits for their completion.     The CLI of PyQBench has a nested structure which is carefully described in \cite{jalowiecki2023pyqbench} for benchmarking based on the discrimination scheme of von Neumann measurements.
    
  Recall that the general form of the CLI invocation is shown in Listing \ref{lst:cli}. 
    \begin{lstlisting}[caption=Invocation of \texttt{qbench} script, label=lst:cli]
  qbench <benchmark-type> <command> <parameters>
  \end{lstlisting}
  So far, PyQBench's CLI has supported only one type of benchmark (based on the discrimination scheme of the parametrized Fourier family of qubit von Neumann measurements). Now, we extend the CLI to the certification scheme.  
  Thus, the accepted values of \texttt{<benchmark-type>} are \texttt{disc-fourier} and \texttt{cert-fourier}. Both have four subcommands: 
  \begin{enumerate}
      \item \texttt{benchmark} -- Running the benchmark. Depending on the experiment scenario, execution can be synchronous, or asynchronous;
       \item \texttt{status} -- Checking the status of the submitted jobs if the execution is asynchronous;
       \item \texttt{resolve} -- Resolving asynchronous jobs into the actual measurement outcomes; 
       \item \texttt{tabulate} -- Converting obtained measurement outcomes into tabulated form in CSV file. 
  \end{enumerate}

\subsubsection{Preparing configuration files}
Now we present how to prepare the configuration files. 
 The first configuration file
describes the experiment scenario to be executed.
\begin{lstlisting}[language=Python, caption=Defining the experiment, label=lst:experiment]
type: certification-fourier
qubits:
	- target: 0
	  ancilla: 1
angles:
	start: 0
	stop: 2 * pi
	num_steps: 8
delta: 0.05
gateset: ibmq
method: direct_sum
num_shots: 10000
\end{lstlisting}

The experiment file contains the following fields:
\begin{itemize}
	\item \texttt{type}: a string describing the type of experiment. Currently,  we command two typos: \texttt{discrimination-fourier} and \texttt{certification-fourier}.
	\item \texttt{qubits}: a list enumerating pairs of qubits on which the experiment should be run. We describe a particular pair of qubits using \texttt{target} and \texttt{ancilla} keys  to emphasize that the role of qubits in the experiment is distinguished.  For the configuration in Listing \ref{lst:experiment}, the benchmark will run on one pair of qubits. In our case, the \texttt{target} is given for qubit 0, and the \texttt{ancilla} is given for qubit~1. 
	\item \texttt{angles}: an object describing the range of angles for the parameterized Fourier family of measurements. The range is always uniform, starts at \texttt{start}, ends at \texttt{stop} and contains \texttt{num\_steps} points, including both \texttt{start} and \texttt{stop}. The \texttt{start} and \texttt{stop} can be arithmetic expressions using \texttt{pi} literal. For example, the range defined in Listing \ref{lst:experiment} contains eight points: $\frac{k\pi}{4}$, where $k=0,\ldots,8$.
    \item  \texttt{delta}:  a float defining the statistical significance chosen from the interval $[0,1]$.  
 	 \item \texttt{gateset}: a string describing the set of gates used in the decomposition of circuits in the experiment. The current version of PyQBench contains explicit implementations of circuits using native gates of \texttt{ibmq} corresponding to decompositions compatible with IBM Q devices. Alternatively, one might wish to turn off the decomposition using a special value \texttt{generic}. 
     \item \texttt{method}: a string, either \texttt{postselection} or \texttt{direct\_sum} determining
which implementation of the conditional measurement is used.
\item \texttt{num\_shots}: an integer defining number of shots are performed in the experiment for a particular angle, qubit pair, and method. 
\end{itemize}

 The second configuration file describes the backend.  Below we describe an example YAML file describing IBM Q backend named Kyiv. Nevertheless, the syntax for each IBM Q backend will be similar.  Note that IBM Q backends typically require an access token to IBM Quantum Experience. The token is configured in \texttt{QISKIT\_IBM\_TOKEN} environmental variable.
\begin{lstlisting}[language=Python, caption=Defining IBMQ backend, label=lst:backend]
name: ibm_kyiv
asynchronous: false
provider:
	hub: ibm-q
	group: open
	project: main
\end{lstlisting}
In this configuration, we have chosen the synchronous mode. Hence, we will first describe considerations related to this mode. We will then briefly describe the asynchronous mode and the differences between these two approaches.

\subsubsection{Running the experiment with stdout output file}
After preparing the YAML files, the benchmark can be launched by using the following command line invocation:
\begin{lstlisting}[language=Python, caption=Running the experiment]
qbench cert-fourier benchmark experiment.yml backend.yml
\end{lstlisting}
The output file will be printed as stdout. Optionally, the \texttt{- -output OUTPUT} parameter might be provided to write the output to the \texttt{OUTPUT} file instead.
\begin{lstlisting}[language=Python, caption=Running the experiment with YML output file]
qbench cert-fourier benchmark experiment.yml backend.yml 
-output async_results.yml
\end{lstlisting}

\subsubsection{Running experiment in the synchronous mode}
If the backend is synchronous, the output will contain the metadata describing the experiment and backend configuration, as well as the histograms with measurement outcomes (bitstrings) for each of the circuits run. The part of output looks similar to the one below. The whole YAML file can be seen in Appendix \ref{app:yml}. 
\begin{lstlisting}[language=Python, caption=YML output file of synchronous experiment]
metadata:
  experiments:
    type: certification-fourier
    qubits:
    - {target: 0, ancilla: 1}
    angles: {start: 0.0, stop: 6.283185307179586, num_steps: 8}
    delta: 0.05
    gateset: ibmq
    method: direct_sum
    num_shots: 10000
  backend_description:
    name: ibm_kyiv
    asynchronous: false
    provider: {group: open, hub: ibm-q, project: main}
data:
- target: 0
  ancilla: 1
  phi: 0.0
  delta: 0.05
  results_per_circuit:
  - name: u
    histogram: {'00': 4787, '01': 4663, '11': 314, '10': 236}
    mitigation_info:
      target: {prob_meas0_prep1: 0.00539999999999996, 
            prob_meas1_prep0: 0.0018}
      ancilla: {prob_meas0_prep1: 0.0048000000000000265, 
            prob_meas1_prep0: 0.0018}
    mitigated_histogram: {'10': 0.02272105079666005, 
        '11': 0.03083240199174341, '01': 0.46866397190027914,
        '00': 0.47778257531131735}\end{lstlisting}
The data whereas includes \texttt{target, ancilla, phi, delta} and \texttt{results\_per\_circuit} information. The
first four pieces of information have already been described. The last data \texttt{results\_per\_circuit} gives us the following additional information:
\begin{itemize}
    \item \texttt{histogram}: the dictionary with measurements' outcomes. The keys represent possible bitstrings, whereas the values are the number of occurrences.
    \item  \texttt{mitigation\_info}: the parameters 
    \texttt{prob\_meas0\_prep1} and \texttt{prob\_meas1\_prep0}
 contains information that are used for error mitigation using the MThree method \cite{mthree} and can be found using \texttt{backends.properties().qubits}. If this information is available, it will be stored in the mitigation info field; otherwise, this field will be absent.
        \item  \texttt{mitigation\_histogram}: the histogram with measurements' outcomes after the error mitigation. 

\end{itemize}

\subsubsection{Running experiment in the asynchronous mode}
For the asynchronous case, we use the same experiment file and we slightly change the backend file  by changing the flag  \texttt{asynchronous:true}. If the backend is asynchronous, the output will contain intermediate data containing, amongst others, \texttt{job\_ids} correlated with the circuit they correspond to. Below, we present YML output file of the asynchronous experiment.
\begin{lstlisting}[language=Python, caption=YML output file of asynchronous experiment]
metadata:
  experiments:
    type: certification-fourier
    qubits:
    - {target: 0, ancilla: 1}
    angles: {start: 0.0, stop: 6.283185307179586, num_steps: 8}
    delta: 0.05
    gateset: ibmq
    method: direct_sum
    num_shots: 10000
  backend_description:
    name: ibm_kyiv
    asynchronous: true
    provider: {group: open, hub: ibm-q, project: main}
data:
- job_id: cy60ptkcw2k0008jwsag
  keys:
  - [0, 1, u, 0.0, 0.05]
  - [0, 1, u, 0.8975979010256552, 0.05]
  - [0, 1, u, 1.7951958020513104, 0.05]
  - [0, 1, u, 2.6927937030769655, 0.05]
  - [0, 1, u, 3.5903916041026207, 0.05]
  - [0, 1, u, 4.487989505128276, 0.05]
  - [0, 1, u, 5.385587406153931, 0.05]
  - [0, 1, u, 6.283185307179586, 0.05]
\end{lstlisting}

\subsubsection{Getting status  of asynchronous jobs}
PyQBench provides also an optional helper command to check the status of asynchronous jobs. The output is a dictionary with histogram of statuses. 
\begin{lstlisting}[language=Python, caption=Getting status of jobs]
qbench cert-fourier status async_results.yml
\end{lstlisting}

\subsubsection{Resolving asynchronous jobs}
If status will be \texttt{DONE}, the stored intermediate data can be resolved in measurements' outcomes.  The resolved results, stored in \texttt{resolved.yml}, would look similar to if the
experiment was run synchronously.  The  YAML output  file can be seen in Appendix \ref{app:yml}.   The following command obtains the resulting file. 
\begin{lstlisting}[language=Python, caption=Resolving status of jobs]
qbench cert-fourier resolve async_results.yml resolved.yml
\end{lstlisting}

\subsubsection{Computing probabilities and tabulating}
No matter in which mode the benchmark was run, the both cases the final output file is suitable for being an input for the command computing the certification probabilities.
As a last step in the processing workflow, the results can be obtained using the command:
\begin{lstlisting}[language=Python, caption=Tabulating results]
qbench cert-fourier tabulate resolved.yml results.csv 
\end{lstlisting}
The results for different certification configurations are discussed in-depth in Section \ref{sec:results}.


\section{Results}\label{sec:results}
In this section, we present the benchmarking results obtained using PyQBench, evaluating the performance of gate-based quantum computers with respect to the certification of qubit von Neumann measurements. The results include experimental outcomes, statistical comparisons, and an analysis of the impact of error mitigation techniques. In addition, we provide a discussion on potential sources of error, limitations of the benchmarking framework, and implications for future research.

The certification scheme was tested on IBM Q devices using the Qiskit library, with experiments conducted for different measurements in Fourier basis. The primary metric used in this benchmarking is the probability of type II error, which measures the likelihood of failing to reject the null hypothesis when the alternative hypothesis is true. Theoretical and experimental values were compared and error mitigation techniques were applied to assess their impact on the results.

To ensure statistical significance, each experimental configuration was repeated multiple times to mitigate the effects of quantum hardware fluctuations and stochastic noise sources. The obtained results were analyzed by computing confidence intervals, which provide an estimate of the experimental uncertainty.

The results given in Table \ref{tab:benchmark_results}, where theoretical and experimental comparisons are based on three key probabilities. The \textit{ideal probability} represents the expected theoretical probability of type II error, assuming perfect quantum measurement fidelity. The \textit{certification probability}  is the observed experimental probability derived from real quantum hardware, while the \textit{mitigated certification probability} incorporates the Mthree error mitigation technique\,\cite{mthree, mthreepublication}  to reduce readout errors and enhance result reliability. Confidence intervals show that the experimental deviations from theoretical predictions remain within statistically acceptable bounds.

\begin{table}[h]
\centering
\begin{tabular}{|c|c|c|c|c|c|c|}
\hline
Target & Ancilla & $\phi$ & $\delta$ & Ideal Prob & Cert Prob & Mitigated Cert Prob \\
\hline
0 & 1 & 0 & 0.05 & 0.95 & 0.948 $\pm$ 0.002 & 0.949 $\pm$ 0.002 \\
0 & 1 & 0.898 & 0.05 & 0.61 & 0.597 $\pm$ 0.003 & 0.597 $\pm$ 0.003 \\
0 & 1 & 1.795 & 0.05 & 0.187 & 0.189 $\pm$ 0.002 & 0.186 $\pm$ 0.002 \\
0 & 1 & 2.696 & 0.05 & 0.0 & 0.0185 $\pm$ 0.001 & 0.0138 $\pm$ 0.001 \\
0 & 1 & 3.59 & 0.05 & 0.0 & 0.021 $\pm$ 0.001 & 0.0163 $\pm$ 0.001 \\
0 & 1 & 4.488 & 0.05 & 0.187 & 0.192 $\pm$ 0.002 & 0.189 $\pm$ 0.002 \\
0 & 1 & 5.386 & 0.05 & 0.61 & 0.611 $\pm$ 0.003 & 0.610 $\pm$ 0.003 \\
0 & 1 & 6.283 & 0.05 & 0.95 & 0.946 $\pm$ 0.002 & 0.949 $\pm$ 0.002 \\
\hline
\end{tabular}
\caption{Benchmarking results for PyQBench experiments, including theoretical probabilities of type II error and experimental values with standard deviation estimates. The tuple  (\textit{Target}, \textit{Ancilla}) describes qubits used during experiment.
The angle $\phi$ determines the Fourier measurement which would like to be certified whereas the parameter $\delta$ refers to the statistical significance chosen in the experiment.
The reference theoretical minimized probability of type II error is presented in \textit{Ideal Prob} column, whereas the obtained,
	empirical probability of type II error is given in \textit{Cert Prob} column. The
	\textit{Mitigated Cert Prob} column contains empirical probability of type II error after applying the
	\texttt{Mthree} error mitigation.}
\label{tab:benchmark_results}
\end{table}

The empirical results closely align with the theoretical predictions, confirming the accuracy of PyQBench in benchmarking the certification of quantum measurements. The results demonstrate that the probability of type II error follows the expected trend across different values of the parameterized Fourier family of measurements. The small deviations observed between theoretical and experimental values arise due to a combination of quantum noise, decoherence effects, and readout errors inherent in the IBM Q hardware.


Error mitigation techniques were applied to improve the accuracy of the experimental results. The Mthree mitigation technique significantly reduced the deviations between the experimental and theoretical results, particularly for measurements where noise effects were more pronounced. The most substantial improvements were observed for small probabilities of type II error, demonstrating that error mitigation can enhance the reliability of certification schemes in quantum benchmarking. The effectiveness of error mitigation highlights the need for integrating noise-aware benchmarking methodologies in practical quantum computing applications.


The statistical significance $\delta$ was fixed at 0.05 for all experiments, ensuring that the probability of type I error remained within acceptable bounds. The experimental results indicate that, for certain values of $\phi$, the empirical probability of type II error deviates slightly from the theoretical predictions. However, the observed discrepancies remain within an acceptable range, validating the robustness of the PyQBench benchmarking framework.






The results are also visualized through three plots that examine the ideal, certified, and mitigated certified probabilities as functions of the parameter $\phi$. These visualizations provide a comprehensive understanding of the certification process, as well as the improvements introduced by the mitigation technique applied on top of it.

\begin{figure}[h]
    \centering
    \includegraphics[width=.8\textwidth]{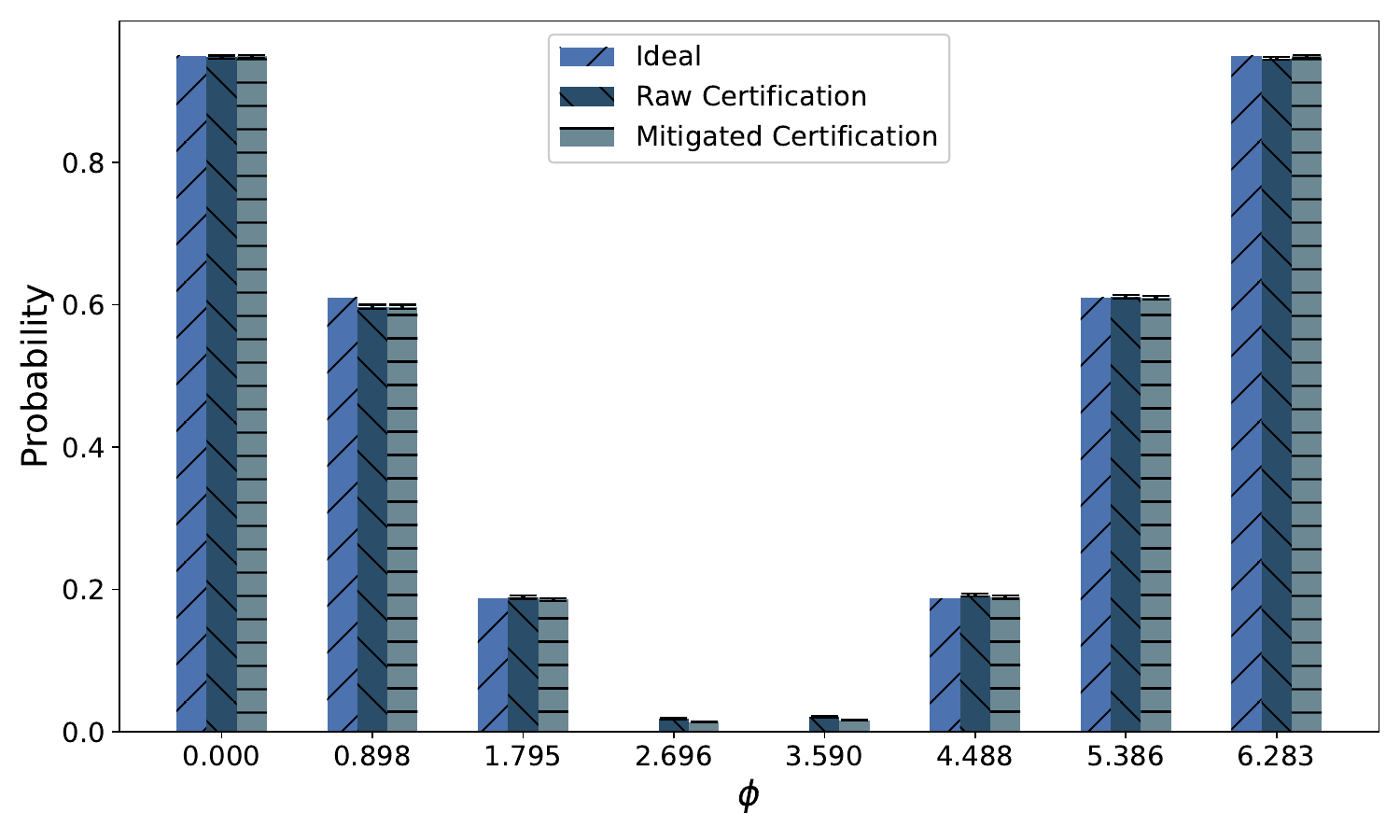}
    \caption{Comparison of probabilities obtained via raw certification and certification with mitigation w.r.t. the theoretical predictions.}
    \label{fig:certProbs}
\end{figure}

 Fig. \ref{fig:certProbs} compares the ideal probability, the certified probability, and the mitigated certified probability for values of $\phi$ ranging from 0 to 6.283. Both the certification and mitigation processes closely track the ideal probability, though the mitigation provides slight improvements over the certified results. For instance, at $\phi = 0$, the ideal probability is 0.95, with the certified probability at 0.948 and the mitigated certified probability at 0.949, each with a margin of error of $\pm 0.002$. However, more noticeable discrepancies arise near $\phi = 2.696$, where the ideal probability is zero. Here, the certified probability reaches 0.0185, whereas the mitigated process reduces this to 0.0138, demonstrating the effectiveness of mitigation in correcting errors introduced by the certification process.

\begin{figure}[h]
    \centering
    \includegraphics[width=.8\textwidth]{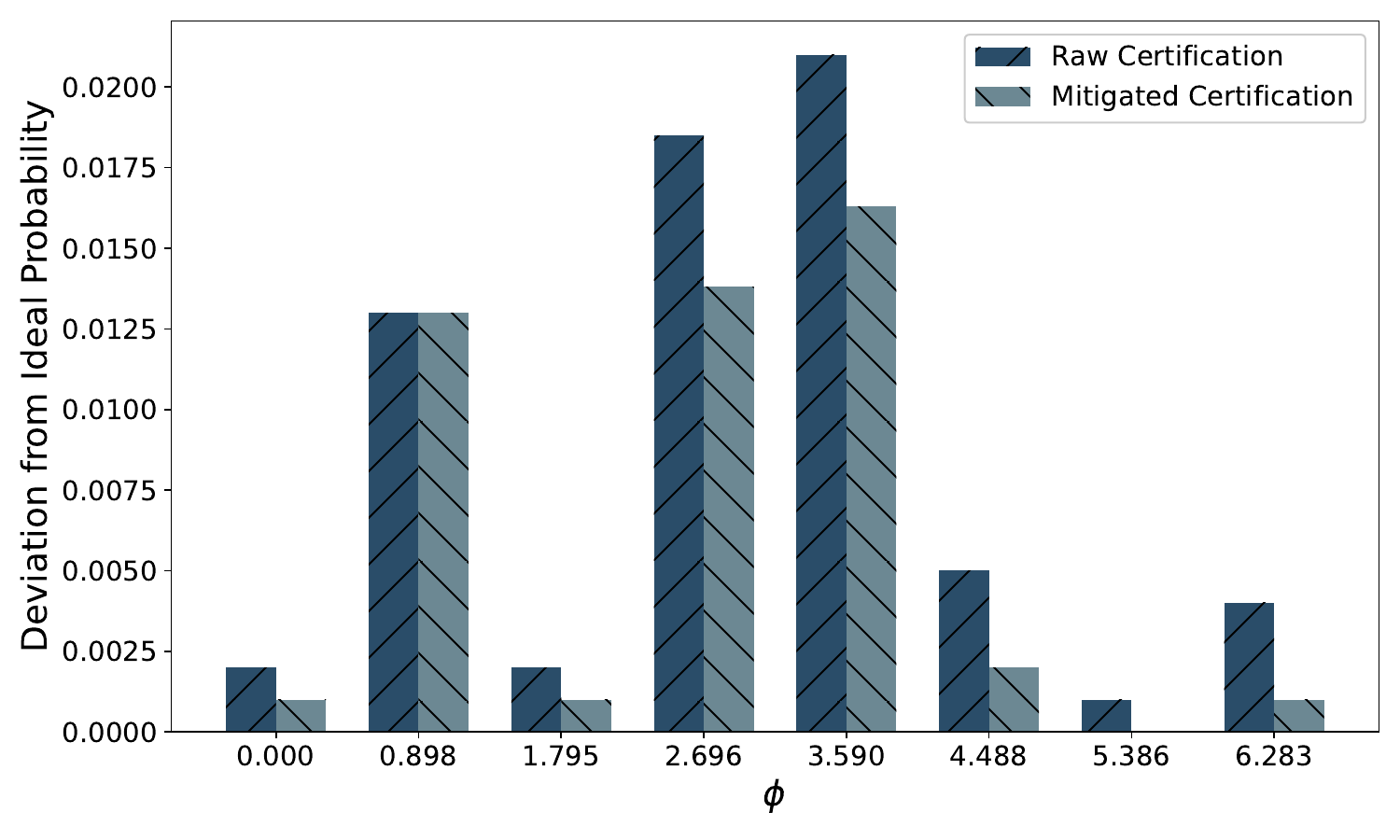}
    \caption{Absolute deviations of the certification both with and without the mitigation from the theoretical probability predictions.}
    \label{fig:certDevs}
\end{figure}

The Fig. \ref{fig:certDevs} illustrates the absolute deviation of both the certified and mitigated certified probabilities from the ideal probability. This plot clearly shows that mitigation reduces the deviation across all values of $\phi$, especially where the certified process shows significant discrepancies. For instance, at $\phi = 2.696$, the certified probability deviates from the ideal by 0.0185, while the mitigated process reduces this to 0.0138. A similar improvement is seen at $\phi = 1.795$, where the certified probability deviates by 0.002 and the mitigated process reduces this deviation to 0.001. Overall, the results indicate that while the certification process provides a reasonable approximation of the ideal, the mitigation technique consistently improves accuracy by reducing deviations, particularly when the certified probability alone deviates significantly.

\begin{figure}[h]
    \centering
    \includegraphics[width=.8\textwidth]{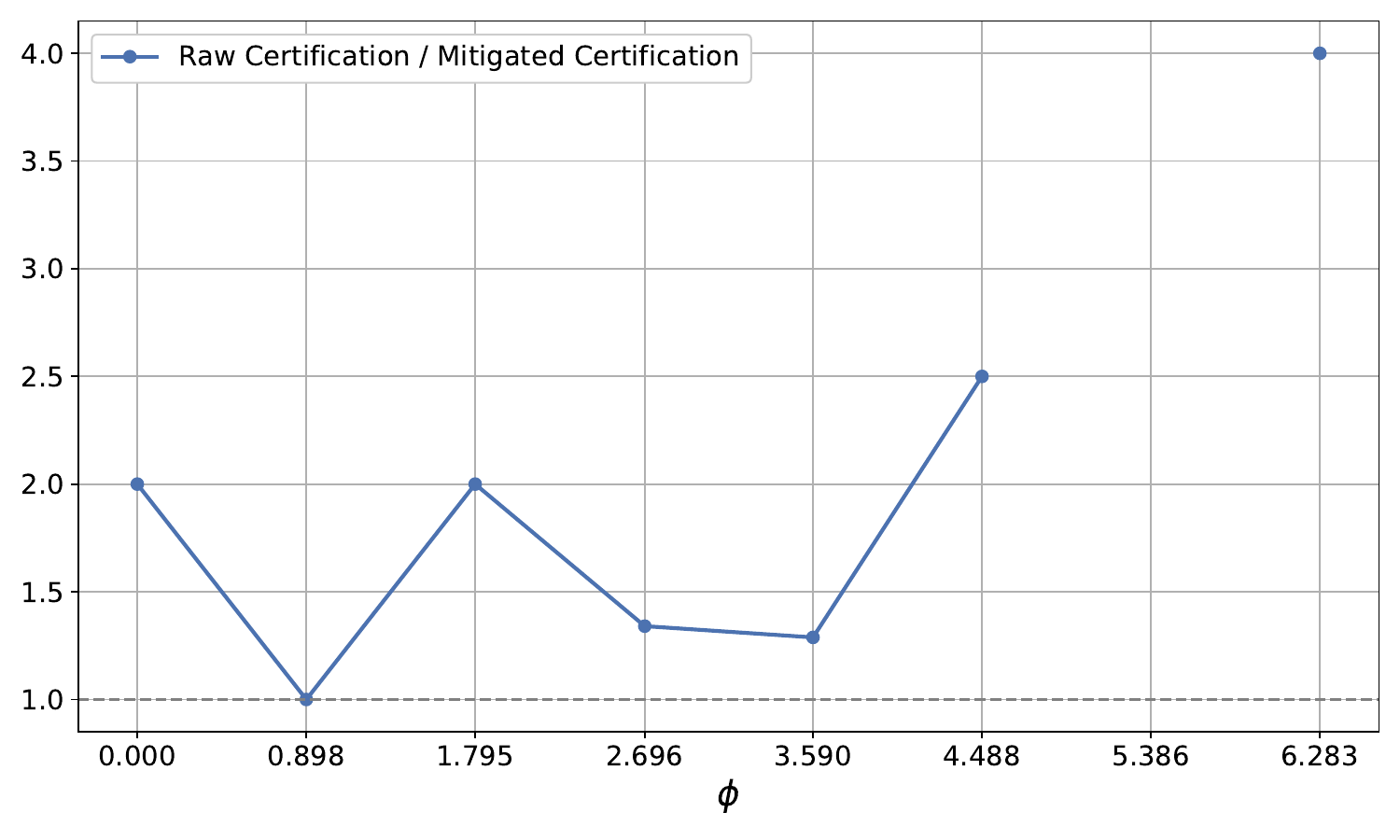}
    \caption{The ratio of deviations of certification both with and without the noise mitigation.}
    \label{fig:rationCertProbs}
\end{figure}

The Fig. \ref{fig:rationCertProbs} explores the ratio between the certified probability and the mitigated certified probability for each value of $\phi$, which provides a clearer view of the impact of mitigation on the certification process. In this context, a ratio of one indicates that both processes yield the same result, as can also be seen at $\phi = 0.898$, while a ratio greater than one demonstrates that mitigation provides an improvement. Across the rest values of $\phi$, the ratio is above one, signifying that mitigation enhances the certified results. The most notable improvement occurs at $\phi = 6.283$, where the ratio reaches 4, indicating that the mitigation process significantly reduces the error introduced by certification. Moreover, at $\phi=5.386$ there is no point in the plot, due to the fact, that the mitigated probability equals the theoretical one precisely, as can be seen in Fig. \ref{fig:certDevs}, so the ratio would be $\infty$.

In summary, the analysis reveals that while the certification process provides a reasonable approximation of the ideal probability, the mitigation technique applied on top of certification consistently enhances the results. The most substantial improvements occur in regions where the ideal probability approaches zero, where the certified process shows more significant deviations from the ideal. The ratio of certified to mitigated certified probabilities further emphasizes the effectiveness of mitigation, as the ratio is generally above one, indicating that the mitigation process outperforms certification alone. This demonstrates the overall effectiveness of the mitigation technique in refining the certification process, particularly when the uncertified results deviate most from the ideal.

The benchmarking experiments using PyQBench validate the effectiveness of the certification scheme for qubit von Neumann measurements. The empirical probability of type II error closely aligns with theoretical values, confirming the accuracy of the methodology. Error mitigation, particularly the Mthree technique, improves result reliability by reducing noise-induced deviations. The mean absolute error is consistently below 0.01, reinforcing the robustness of the framework. Thus, PyQBench offers a scalable and statistically rigorous approach for benchmarking quantum measurement accuracy on IBM Q devices.

These results provide insights into the performance of quantum hardware and reinforce the importance of certification-based benchmarking techniques in quantum computing research. Future work should explore the application of PyQBench to alternative quantum computing platforms, as well as the integration of additional error mitigation strategies to further enhance benchmarking accuracy.

\section{Conclusions} \label{sec:conclusion}


In this paper, we introduced an updated version of PyQBench, an open-source Python library designed to benchmark gate-based quantum computers. This version extends the original PyQBench framework by incorporating a certification scheme for qubit von Neumann measurements, thereby providing a more comprehensive method for evaluating Noisy Intermediate-Scale Quantum (NISQ) devices. The new functionalities include tools for certifying the accuracy of quantum measurements, in addition to discriminating them, and have been specifically designed to integrate seamlessly with IBM Q devices via the Qiskit library.

The enhanced PyQBench tool fills a crucial gap in the quantum computing landscape, where the need for reliable benchmarking methodologies is becoming increasingly important as NISQ devices mature. By supporting customizable measurement schemes and error models, PyQBench enables a more detailed and flexible analysis of quantum hardware, offering valuable insights into device performance under real-world conditions. Moreover, its user-friendly command-line interface (CLI) and Python library integration make it accessible to both novice users and advanced researchers, encouraging widespread adoption and further development.

Beyond its technical contributions, PyQBench also plays a significant role in advancing the field of quantum hardware validation. The integration of the qubit von Neumann measurement certification scheme allows researchers to more accurately assess the fidelity of quantum measurements, which is critical for the development of error mitigation techniques and quantum error correction in NISQ systems. This work aligns with broader efforts in quantum information science to establish reliable performance metrics for quantum devices, helping bridge the gap between theoretical advancements and practical implementations.

Future work on PyQBench could focus on expanding its applicability to a wider range of quantum hardware platforms beyond IBM Q, such as Google's Sycamore or Honeywell's trapped-ion systems. Additionally, integrating more sophisticated error mitigation techniques and benchmarking metrics, such as quantum volume or cross-entropy benchmarking, would further enhance its utility. These directions not only hold promise for improving the benchmarking of quantum hardware but also contribute to the overall goal of making quantum computing more robust and scalable in real-world applications.

In conclusion, the new version of PyQBench provides a powerful, flexible, and open-source tool for the quantum computing community. Its contributions to quantum hardware benchmarking, particularly through the certification of quantum measurements, will be valuable as researchers continue to push the boundaries of what NISQ devices can achieve. By making this tool freely available and encouraging community contributions, we hope to foster collaboration and innovation in the development of benchmarking tools for near-term quantum technologies.

This enhanced version of PyQBench has been specifically adapted to benchmark IBM Q quantum devices via the Qiskit library. The source code is freely available on GitHub under the Apache License v2, encouraging users to not only apply the tool to their specific research needs but also contribute to and extend its functionality for broader applications.

\section*{Acknowledgements}

It is a pleasure to thank \L{}ukasz Pawela for numerous discussions concerning experiments on IBM quantum devices. 

PL is supported by the Ministry of Education, Youth and Sports of the Czech Republic through the e-INFRA CZ (ID:90254),
with the financial support of the European Union under the REFRESH - Research Excellence For REgion Sustainability and High-tech Industries project number CZ.10.03.01/00/22\_003/0000048 via the Operational Programme Just Transition.
MB is supported by Italian Government (Ministero dell'Universit\`a e della Ricerca, PRIN 2022 PNRR) -
cod. P2022SELA7: ``RECHARGE: monitoRing, tEsting, and
CHaracterization of performAnce Regressions'' - D.D. n. 1205
del 28/7/2023.

\section*{References}
\bibliographystyle{unsrt}
\bibliography{benchmarking}

\appendix

\section{Optimal strategy for parametrized Fourier family of measurements}\label{app:fourier}

 In this Appendix, we construct the optimal theoretical strategy of certification for parametrized Fourier family of qubit measurements. 
 
 The first component is to determine the optimal initial state $\ket{\psi_0}$.  From \cite{lewandowska2021certification} we know that $\ket{\psi_0} $ maximizes the diamond norm between $\PP_U$ and $\PP_\Id$, that is, $\|  \PP_U - \PP_\Id \|_\diamond \coloneqq \max_{\| \ket{\psi} \|_1 = 1 } \|((\PP_U - \PP_\Id) \otimes \Id) (\proj{\psi})\|_1$. Then, from \cite{jalowiecki2023pyqbench} (see Proposition 1 in Appendix D), we immediately see that $\ket{\psi_0}$ is the Bell state.    
 The next step is to determine the final optimal measurement $\Omega_0$. 
 For this purpose, we present the following proposition.

 \begin{proposition}\label{bench-cert-strategy-measurement}
 	For two-point certification of von Neumann measurements $\PP_\Id$ and $\PP_{U}$  defined in Eq.~\eqref{fourier-cert-bench} with
 	statistical significance $\delta$, the
 	controlled unitaries $V_0$ and $V_1$ have the form 
 	\begin{enumerate}
 		\item if $\sqrt{1+\cos\phi} \ge \sqrt{2\delta} $ and $\phi \in [0, \pi)$, then  \begin{equation}
 		V_0 = \left(\begin{array}{cc}\sqrt{1-\delta}&\sqrt{\delta}\\-\sqrt{\delta}&\sqrt{1-\delta}\end{array}\right),
 		\end{equation} 
 		and 
 		\begin{equation}
 		V_1 = \left(\begin{array}{cc}\sqrt{\delta}&\sqrt{1-\delta}\\\sqrt{1-\delta}&-\sqrt{\delta}\end{array}\right);
 		\end{equation} 
 		
 		\item if $ \sqrt{1+\cos\phi} < \sqrt{2\delta} $ and $\phi \in [0, 2\pi)$, then  \begin{equation}
 		V_0 = \left(\begin{array}{cc}\sin \frac{\phi}{2}& | \cos \frac{\phi}{2}|\\-\cos \frac{\phi}{2}&\frac{\sin \phi}{2 |\cos \frac{\phi}{2}| }\end{array}\right),
 		\end{equation}
 		and 
 		\begin{equation}
 		V_1 = \left(\begin{array}{cc} | \cos \frac{\phi}{2} | &  \sin \frac{\phi}{2} \\ \frac{\sin\phi}{2|\cos \frac{\phi}{2}|} & - \cos \frac{\phi}{2} \end{array}\right); 
 		\end{equation}
 		
 		\item $\sqrt{1+\cos\phi} \ge \sqrt{2\delta}$ and $\phi \in [\pi, 2\pi)$, then  \begin{equation}
 		V_0 = \left(\begin{array}{cc}\sqrt{1-\delta}&-\sqrt{\delta}\\\sqrt{\delta}&\sqrt{1-\delta}\end{array}\right),
 		\end{equation}
 		and 
 		\begin{equation}
 		V_1= \left(\begin{array}{cc} - \sqrt{\delta}&\sqrt{1-\delta}\\\sqrt{1-\delta}&\sqrt{\delta}\end{array}\right). 
 		\end{equation}
 	\end{enumerate} 
 \end{proposition} 
 
 \begin{proof} To determine the controlled unitaries $V_0$ and $V_1$, we first use 
Algorithm 1 from \cite{lewandowska2021certification}. From the  Proposition 1 in \cite{jalowiecki2023pyqbench}, we assume $\ket{\psi_0} = \frac{1}{\sqrt{2}} \left( \ket{00} + \ket{11} \right)$. Then, we perform the measurement $\PP $ being $\PP_U$ or $\PP_\Id$ on the first subsystem, so we have the state conditioned by label $i$ given by $\proj{\psi_i} \propto (\bra{i} \otimes \Id) (\PP \otimes \Id) \proj{\psi_0} (\ket{i} \otimes \Id) $, where $i=0,1$. Then, we  assume two cases. 

 	For the case $i=0$, the certification scheme of von Neumann measurements is reduced to the certification scheme of the following quantum states:
 	\begin{eqnarray}
 	&H_0: \  \ket{0} \\
 	&H_1: \   H 
 	\left(\begin{array}{cc}1&0\\0&e^{i \phi}\end{array}\right)  H^\dagger \ket{0 } = \frac{1 + e^{-i \phi}}{2} \ket{0} + \frac{ 1 -  e^{-i \phi}}{2}\ket{1}
 	\end{eqnarray}
 	Next, we will use the the proof of Theorem 1 from \cite{lewandowska2021certification}. 
 	According to the assumption $\ket{\psi} = \alpha \ket{\phi} + \beta\ket{\phi^\perp}$, for some $\alpha,\beta
 	\geq0$ satisfying $\alpha^2+\beta^2=1$, we need to rewrite the state  $\frac{1}{2}\left(1+ e^{-i \phi}\right) \ket{0} + \frac{1}{2}\left(1- e^{-i \phi}\right) \ket{1} $
 	to the form 
 	$ \frac{\sqrt{1+\cos\phi}}{\sqrt{2}}  \ket{0} + \frac{i \, \sin\phi} {\sqrt{2+2\cos\phi}}  \ket{1}$.
 	Hence, we obtain the following hypothesis
 	\begin{eqnarray}
 	&H_0: \  \ket{0} \\
 	&H_1: \frac{\sqrt{1+\cos\phi}}{\sqrt{2}}  \ket{0} + \frac{i \, \sin\phi} {\sqrt{2+2\cos\phi}}  \ket{1}.
 	\end{eqnarray}
 	To achieve the optimal final measurement for such defined certification scheme, we use Corollary 1 from \cite{lewandowska2021certification}.
    The inner product between states $\ket{0} $ and $\frac{\sqrt{1+\cos\phi}}{\sqrt{2}}  \ket{0} + \frac{i \, \sin\phi} {\sqrt{2+2\cos\phi}}  \ket{1}$ equals $\frac{\sqrt{1 + \cos \phi}}{\sqrt{2}}$ and then we have:
 	
 	
 	\begin{enumerate}
 		\item if $\sqrt{1+\cos\phi} \ge \sqrt{2\delta} $ and $\phi \in [0, \pi)$, then the optimal
 		measurement is given by   $\Omega_0= \ketbra{\omega}{\omega}$, where 
 		\begin{equation}
 		\ket{\omega} = \sqrt{1-\delta} \ket{0} - \sqrt{\delta} \ket{1},
 		\end{equation}
 		\item if $ \sqrt{1+\cos\phi} < \sqrt{2\delta} $ and $\phi \in [0, 2\pi)$, then the optimal
 		measurement is given by   $\Omega_0= \ketbra{\omega}{\omega}$, where 
 		\begin{equation}
 		\ket{\omega} = \sin \frac{\phi}{2} \ket{0} - \cos \frac{\phi}{2} \ket{1},
 		\end{equation}
 		\item $\sqrt{1+\cos\phi} \ge \sqrt{2\delta}$ and $\phi \in [\pi, 2\pi)$, then the optimal
 		measurement is given by   $\Omega_0= \ketbra{\omega}{\omega}$, where 
 		\begin{equation}
 		\ket{\omega} = \sqrt{1-\delta} \ket{0} +\sqrt{ \delta} \ket{1},
 		\end{equation}
 	\end{enumerate} 
 	which complete first part of the proof.

 	For the case $i=1$, whereas, the certification scheme of von Neumann measurements is reduced to the certification scheme of quantum states:
 	\begin{eqnarray}
 	&H_0: \  \ket{1} \\
 	&H_1: \   H 
 	\left(\begin{array}{cc}1&0\\0&e^{i \phi}\end{array}\right)  H^\dagger \ket{1 } =   \frac{1- e^{-i \phi}}{2} \ket{0} + \frac{ 1 +  e^{-i \phi}}{2}\ket{1}
 	\end{eqnarray} 
 	In the same way, we rewrite the state $\frac{1}{2}\left(1- e^{-i \phi}\right) \ket{0} + \frac{1}{2}\left(1 + e^{-i \phi}\right) \ket{1} $
 	to the form  $\frac{\sin\phi}{\sqrt{2+2\cos\phi}}\ket{0} + \frac{\sqrt{1+\cos\phi}}{\sqrt{2}}\ket{1}$. 
 	Then,  we have the following hypothesis: 
 	\begin{eqnarray}
 	&H_0: \  \ket{1} \\
 	&H_1:  \frac{\sin\phi}{\sqrt{2+2\cos\phi}}\ket{0} + \frac{\sqrt{1+\cos\phi}}{\sqrt{2}}\ket{1}
 	\end{eqnarray}
 	
 	Again, to achieve the optimal final measurement for such defined certification scheme, we use Corollary 1,  and then we assume three cases:  
 	\begin{enumerate}
 		\item if $\sqrt{1+\cos\phi} \ge \sqrt{2\delta} $ and $\phi \in [0, \phi)$, then the optimal
 		measurement is given by   $\Omega_0= \ketbra{\omega}{\omega}$, where 
 		\begin{equation}
 		\ket{\omega} = \sqrt{\delta} \ket{0} + \sqrt{1- \delta} \ket{1}.
 		\end{equation}
 		\item if $ \sqrt{1+\cos\phi} < \sqrt{2\delta} $ and $\phi \in [0, 2\pi)$, then the optimal
 		measurement is given by   $\Omega_0= \ketbra{\omega}{\omega}$, where 
 		\begin{equation}
 		\ket{\omega} = \left| \cos \frac{\phi}{2} \right|  \ket{0} + \frac{\sin\phi}{2 \left|\cos \frac{\phi}{2} \right| }\ket{1}.
 		\end{equation}
 		\item $\sqrt{1+\cos\phi} \ge \sqrt{2\delta}$ and $\phi \in [\pi, 2\pi)$, then the optimal
 		measurement is given by   $\Omega_0= \ketbra{\omega}{\omega}$, where 
 		\begin{equation}
 		\ket{\omega} = -\sqrt{\delta} \ket{0} +\sqrt{ 1-\delta} \ket{1}.
 		\end{equation}
 		
 	\end{enumerate}
 	which complete the second part of the proof. 
 \end{proof} 

Finally, we need to calculate the minimized probability of type II error. 
Observe that for certification scheme of the quantum states:
 	\begin{eqnarray}
 	&H_0: \  \ket{i} \\
 	&H_1: \   H 
 	\left(\begin{array}{cc}1&0\\0&e^{i \phi}\end{array}\right)  H^\dagger \ket{i} = \frac{1 + e^{-i \phi}}{2} \ket{i} + \frac{ 1 -  e^{-i \phi}}{2}\ket{i^\perp}
 	\end{eqnarray}
the inner product equals $\frac{1+e^{-i \phi}}{2}$. It implies, from Theorem 1 \cite{lewandowska2021certification}, that the minimized probability of the type II error yields
	\begin{equation}
	p_{\text{II}}(\phi) =\left\{
	\begin{array}{ccc}
	\left( \frac{ |1+e^{i  \phi} | }{2} \cdot \sqrt{1-\delta} - \sqrt{1-\frac{|1+e^{i  \phi}|^2}{4}}  \cdot \sqrt{\delta } \right)^2 &\mbox{\text{for}}&\frac{ |1+e^{i  \phi} | }{2} > \sqrt{\delta},\\
	0&\mbox{\text{for}}&\frac{ |1+e^{i  \phi} | }{2} \le \sqrt{\delta}.
	\end{array}
	\right.
	\end{equation}

\section{YML output file of synchronous experiment} \label{app:yml}
\begin{lstlisting}[language=Python, caption=YML output file of synchronous experiment]
metadata:
  experiments:
    type: certification-fourier
    qubits:
    - {target: 0, ancilla: 1}
    angles: {start: 0.0, stop: 6.283185307179586, num_steps: 8}
    delta: 0.05
    gateset: ibmq
    method: direct_sum
    num_shots: 10000
  backend_description:
    name: ibm_kyiv
    asynchronous: false
    provider: {group: open, hub: ibm-q, project: main}
data:
- target: 0
  ancilla: 1
  phi: 0.0
  delta: 0.05
  results_per_circuit:
  - name: u
    histogram: {'00': 4787, '01': 4663, '11': 314, '10': 236}
    mitigation_info:
      target: {prob_meas0_prep1: 0.00539999999999996, 
            prob_meas1_prep0: 0.0018}
      ancilla: {prob_meas0_prep1: 0.0048000000000000265, 
            prob_meas1_prep0: 0.0018}
    mitigated_histogram: {'10': 0.02272105079666005, 
            '11': 0.0308324019917434, '01': 0.4686639719002791,
            '00': 0.47778257531131735}
- target: 0
  ancilla: 1
  phi: 0.8975979010256552
  delta: 0.05
  results_per_circuit:
  - name: u
    histogram: {'01': 3158, '11': 1841, '10': 2121, '00': 2880}
    mitigation_info:
      target: {prob_meas0_prep1: 0.00539999999999996,  
            prob_meas1_prep0: 0.0018}
      ancilla: {prob_meas0_prep1: 0.0048000000000000265, 
            prob_meas1_prep0: 0.0018}
    mitigated_histogram: {'11': 0.18503494944149518, 
            '10': 0.2119853847642628, '00': 0.2863022864686138,
            '01': 0.3166773793256281}
- target: 0
  ancilla: 1
  phi: 1.7951958020513104
  delta: 0.05
  results_per_circuit:
  - name: u
    histogram: {'11': 3946, '10': 4107, '01': 933, '00': 1014}
    mitigation_info:
      target: {prob_meas0_prep1: 0.00539999999999996, 
            prob_meas1_prep0: 0.0018}
      ancilla: {prob_meas0_prep1: 0.0048000000000000265, 
            prob_meas1_prep0: 0.0018}
    mitigated_histogram: {'01': 0.09187983560151308, 
            '00': 0.0992818314007015, '11': 0.3977454665741516,
            '10': 0.4110928664236337}
- target: 0
  ancilla: 1
  phi: 2.6927937030769655
  delta: 0.05
  results_per_circuit:
  - name: u
    histogram: {'10': 4905, '11': 4844, '00': 121, '01': 130}
    mitigation_info:
      target: {prob_meas0_prep1: 0.00539999999999996, 
            prob_meas1_prep0: 0.0018}
      ancilla: {prob_meas0_prep1: 0.0048000000000000265, 
            prob_meas1_prep0: 0.0018}
    mitigated_histogram: {'00': 0.00971145661283892, 
            '01': 0.0107234135301044, '11': 0.4884707846971315,
            '10': 0.491094345159925}
- target: 0
  ancilla: 1
  phi: 3.5903916041026207
  delta: 0.05
  results_per_circuit:
  - name: u
    histogram: {'11': 4727, '10': 5047, '01': 119, '00': 107}
    mitigation_info:
      target: {prob_meas0_prep1: 0.00539999999999996, 
            prob_meas1_prep0: 0.0018}
      ancilla: {prob_meas0_prep1: 0.0048000000000000265, 
            prob_meas1_prep0: 0.0018}
    mitigated_histogram: {'00': 0.008243326178401936, 
            '01': 0.0096749343410262, '11': 0.4766264355219874, 
            '10': 0.5054553039585844}
- target: 0
  ancilla: 1
  phi: 4.487989505128276
  delta: 0.05
  results_per_circuit:
  - name: u
    histogram: {'10': 4067, '00': 924, '01': 1045, '11': 3964}
    mitigation_info:
      target: {prob_meas0_prep1: 0.00539999999999996, 
            prob_meas1_prep0: 0.0018}
      ancilla: {prob_meas0_prep1: 0.0048000000000000265, 
            prob_meas1_prep0: 0.0018}
    mitigated_histogram: {'00': 0.09020755763218545, 
            '01': 0.103168725838722, '11': 0.39955085514435557,
            '10': 0.4070728613847364}
- target: 0
  ancilla: 1
  phi: 5.385587406153931
  delta: 0.05
  results_per_circuit:
  - name: u
    histogram: {'10': 1973, '11': 1989, '00': 3069, '01': 2969}
    mitigation_info:
      target: {prob_meas0_prep1: 0.00539999999999996, 
            prob_meas1_prep0: 0.0018}
      ancilla: {prob_meas0_prep1: 0.0048000000000000265, 
            prob_meas1_prep0: 0.0018}
    mitigated_histogram: {'10': 0.1969715269854367, 
            '11': 0.2000488072203213, '01': 0.2975337874613869,
            '00': 0.305445878332855}
- target: 0
  ancilla: 1
  phi: 6.283185307179586
  delta: 0.05
  results_per_circuit:
  - name: u
    histogram: {'00': 4832, '01': 4661, '11': 257, '10': 250}
    mitigation_info:
      target: {prob_meas0_prep1: 0.00539999999999996, 
            prob_meas1_prep0: 0.0018}
      ancilla: {prob_meas0_prep1: 0.0048000000000000265, 
            prob_meas1_prep0: 0.0018}
    mitigated_histogram: {'10': 0.024153348441373016, 
            '11': 0.0250715357945843, '01': 0.4684820500233045, 
            '00': 0.48229306574073816}\end{lstlisting}
        
\begin{lstlisting}[language=Python, caption=YML output file of asynchronous experiment]
metadata:
  experiments:
    type: certification-fourier
    qubits:
    - target: 0
      ancilla: 1
    angles:
      start: 0.0
      stop: 6.283185307179586
      num_steps: 8
    delta: 0.05
    gateset: ibmq
    method: direct_sum
    num_shots: 10000
  backend_description:
    name: ibm_kyiv
    asynchronous: true
    provider:
      group: open
      hub: ibm-q
      project: main
data:
- target: 0
  ancilla: 1
  phi: 0.0
  delta: 0.05
  results_per_circuit:
  - name: u
    histogram:
      '00': 4898
      '01': 4582
      '11': 311
      '10': 209
    mitigation_info:
      target:
        prob_meas0_prep1: 0.00539999999999996
        prob_meas1_prep0: 0.0018
      ancilla:
        prob_meas0_prep1: 0.0048000000000000265
        prob_meas1_prep0: 0.0018
    mitigated_histogram:
      '10': 0.01998449638687051
      '11': 0.030549024853314708
      '01': 0.4604864304246869
      '00': 0.4889800483351278
- target: 0
  ancilla: 1
  phi: 0.8975979010256552
  delta: 0.05
  results_per_circuit:
  - name: u
    histogram:
      '11': 1952
      '00': 2876
      '10': 2074
      '01': 3098
    mitigation_info:
      target:
        prob_meas0_prep1: 0.00539999999999996
        prob_meas1_prep0: 0.0018
      ancilla:
        prob_meas0_prep1: 0.0048000000000000265
        prob_meas1_prep0: 0.0018
    mitigated_histogram:
      '11': 0.19626874194826874
      '10': 0.2071941128936882
      '00': 0.28595657203781866
      '01': 0.3105805731202244
- target: 0
  ancilla: 1
  phi: 1.7951958020513104
  delta: 0.05
  results_per_circuit:
  - name: u
    histogram:
      '01': 905
      '11': 4022
      '10': 4085
      '00': 988
    mitigation_info:
      target:
        prob_meas0_prep1: 0.00539999999999996
        prob_meas1_prep0: 0.0018
      ancilla:
        prob_meas0_prep1: 0.0048000000000000265
        prob_meas1_prep0: 0.0018
    mitigated_histogram:
      '01': 0.08902728596975329
      '00': 0.09669850424566845
      '11': 0.4054328268424948
      '10': 0.40884138294208333
- target: 0
  ancilla: 1
  phi: 2.6927937030769655
  delta: 0.05
  results_per_circuit:
  - name: u
    histogram:
      '10': 4990
      '11': 4825
      '00': 81
      '01': 104
    mitigation_info:
      target:
        prob_meas0_prep1: 0.00539999999999996
        prob_meas1_prep0: 0.0018
      ancilla:
        prob_meas0_prep1: 0.0048000000000000265
        prob_meas1_prep0: 0.0018
    mitigated_histogram:
      '00': 0.005669915398114243
      '01': 0.008121105338749025
      '11': 0.4865404579166901
      '10': 0.4996685213464466
- target: 0
  ancilla: 1
  phi: 3.5903916041026207
  delta: 0.05
  results_per_circuit:
  - name: u
    histogram:
      '11': 4826
      '10': 4964
      '00': 103
      '01': 107
    mitigation_info:
      target:
        prob_meas0_prep1: 0.00539999999999996
        prob_meas1_prep0: 0.0018
      ancilla:
        prob_meas0_prep1: 0.0048000000000000265
        prob_meas1_prep0: 0.0018
    mitigated_histogram:
      '00': 0.007888851269304064
      '01': 0.008418779091074404
      '11': 0.4866456850507466
      '10': 0.49704668458887485
- target: 0
  ancilla: 1
  phi: 4.487989505128276
  delta: 0.05
  results_per_circuit:
  - name: u
    histogram:
      '11': 3954
      '10': 4124
      '00': 889
      '01': 1033
    mitigation_info:
      target:
        prob_meas0_prep1: 0.00539999999999996
        prob_meas1_prep0: 0.0018
      ancilla:
        prob_meas0_prep1: 0.0048000000000000265
        prob_meas1_prep0: 0.0018
    mitigated_histogram:
      '00': 0.08667377941797003
      '01': 0.10197127796072936
      '11': 0.398532348147248
      '10': 0.41282259447405256
- target: 0
  ancilla: 1
  phi: 5.385587406153931
  delta: 0.05
  results_per_circuit:
  - name: u
    histogram:
      '00': 3129
      '10': 1964
      '01': 2977
      '11': 1930
    mitigation_info:
      target:
        prob_meas0_prep1: 0.00539999999999996
        prob_meas1_prep0: 0.0018
      ancilla:
        prob_meas0_prep1: 0.0048000000000000265
        prob_meas1_prep0: 0.0018
    mitigated_histogram:
      '11': 0.19408825484296274
      '10': 0.19608690118683395
      '01': 0.2983573535373757
      '00': 0.3114674904328276
- target: 0
  ancilla: 1
  phi: 6.283185307179586
  delta: 0.05
  results_per_circuit:
  - name: u
    histogram:
      '01': 4768
      '00': 4696
      '10': 241
      '11': 295
    mitigation_info:
      target:
        prob_meas0_prep1: 0.00539999999999996
        prob_meas1_prep0: 0.0018
      ancilla:
        prob_meas0_prep1: 0.0048000000000000265
        prob_meas1_prep0: 0.0018
    mitigated_histogram:
      '10': 0.02325138419612959
      '11': 0.02889276720310536
      '00': 0.4685898728546358
      '01': 0.4792659757461291
\end{lstlisting}

\end{document}